\documentclass[journal]{IEEEtran}
\ifCLASSINFOpdf
\else
\fi
\usepackage[cmex10]{amsmath}
\usepackage{nonlinfilt}
\usepackage{textcomp}
\usepackage{amsbsy}
\usepackage{latexsym}
\usepackage{amsfonts,amsthm}
\usepackage{amssymb}
\usepackage{caption}
\usepackage{subcaption}
\usepackage[normalem]{ulem}
\usepackage{url}

\def\arXiv#1{arXiv:\href{http://arXiv.org/abs/#1}{#1}}

\def\MR#1{MR\href{http://www.ams.org/mathscinet-getitem?mr=#1}{#1}}

\usepackage{hyperref}
\usepackage{bm}

\newtheorem{conjec}[lemma]{Conjecture}
\definecolor{ao}{rgb}{0.0, 0.5, 0.0}
\definecolor{ID}{rgb}{0.5, 0.5, 0.5}
\definecolor{nd}{rgb}{0.5, 0.5, 0.5}
\def\RR{\rm \hbox{I\kern-.2em\hbox{R}}}
\def\NN{\rm \hbox{I\kern-.2em\hbox{N}}}
\def\ZZ{\rm {{\rm Z}\kern-.28em{\rm Z}}}
\def\CC{\rm \hbox{C\kern -.5em {\raise .32ex \hbox{$\scriptscriptstyle
|$}}\kern
-.22em{\raise .6ex \hbox{$\scriptscriptstyle |$}}\kern .4em}}
\def\vp{\varphi}
\def\<{\langle}
\def\>{\rangle}

\def\Chi{\raise .3ex
\hbox{\large $\chi$}} \def\vp{\varphi}
\def\lsima{\hbox{\kern -.6em\raisebox{-1ex}{$~\stackrel{\textstyle<}{\sim}~$}}\kern -.4em}
\def\lsim{\hbox{\kern -.2em\raisebox{-1ex}{$~\stackrel{\textstyle<}{\sim}~$}}\kern -.2em}
\def\[{\Bigl [}
\def\]{\Bigr ]}
\def\({\Bigl (}
\def\){\Bigr )}
\def\[{\Bigl [}
\def\]{\Bigr ]}
\def\({\Bigl (}
\def\){\Bigr )}

\def\ds{\displaystyle}
\def\R{\mathbb{R}}

\def\T{{\relax\ifmmode I\!\!\hspace{-1pt}T\else$I\!\!\hspace{-1pt}T$\fi}}
\def\C{\mathbb{C}}

\def\lsim{\raisebox{-1ex}{$~\stackrel{\textstyle<}{\sim}~$}}

  \def\NN{N}                  

\newcommand{\be}{\begin{equation}}
\newcommand{\ee}{\end{equation}}
\newcommand{\bea}{$$ \begin{array}{lll}}
\newcommand{\eea}{\end{array} $$}

\def \exp{\mathop{\rm    exp}}

\newcommand{\beqn}{\begin{equation}}
\newcommand{\eeqn}{\end{equation}}

\def\endproof{\hfill\rule{1.5mm}{1.5mm}\\[2mm]}

\makeatletter
\@addtoreset{equation}{section}
\makeatother

\def\int{\intop\limits}

\def\C{\mathbb C}
\def\CP{\mathbb {CP}}

\def\F{\mathbb F}
\def\FP{\mathbb {FP}}

\def\R{\mathbb R}

\def\S{\mathbb S}

\def\theta{\vartheta}

\begin{document}
%
\title{Grassmannian packings: Trust-region stochastic tuning for matrix incoherence}
%
%
%

\author{Josiah~Park,
        Carlos~Saltijeral,
        and~Ming~Zhong \
$${\tt j.park@tamu.edu,\ csaltijeral@tamu.edu,\ mingzhong@tamu.edu}$$ 
\thanks{This research was supported by NSF Tripods Grant CCF-1934904 (JP).}
}

\maketitle

\begin{abstract}
  We provide a new numerical procedure for constructing low coherence matrices, Trust-Region Stochastic Tuning for Matrix Incoherence (TRSTMI) and detail the results of experiments with a CPU/GPU parallelized implementation of this method. These trials suggest the superiority of this approach over other existing methods when the size of the matrix is large. We also present new conjectures on optimal complex matrices motivated and guided by the experimental results.

\end{abstract}


%
\IEEEpeerreviewmaketitle

%
%
%
%

\section{Introduction}

Structured point configurations have generated much interest in recent years in connection with their role in applications to coding theory, wireless beam-forming, and compressed sensing. Some of these configurations, like packings which have maximal separation of elements, arise naturally in communications due their optimal properties for quantization. Multiple-Input Multiple-Output (MIMO) wireless communication is one of the main areas where such configurations (called {\it Grassmannian constellations}) have gained interest due to desirable statistical properties for signal transmission and other domains \cite{ARU,hmr,LHS,TSR,ZWG}.

An optimal packing of lines or subspaces in $\F^d$ ($\F=\R,\C$, or $\mathbb{H}$, however we focus here on $\F=\C$) satisfies that the minimal pairwise {\it chordal distance} is maximized over all choices of collections of $k$-dimensional subspaces in $ \F^{d}$, $\pmb{\Phi}=\{\pmb{\varphi}_{j}\}_{j=1}^N$. The difficulty in designing such codebooks is demonstrated by the extensive literature (see \cite{bgm,cbs,MD,YRP} for instance) addressing attempts to numerically and algebraically construct optimal and near-optimal packings. Conway, Hardin, and Sloane pioneered numerical searches for real subspace packings \cite{CHS} hosting an online database of these configurations \cite{Slo}. 

 For simplicity, we address here only the case that $k=1$, corresponding to the problem of packing lines. In this case, lines may be identified with points in the projective space $\FP^{d-1}$, and so we will call such packings {\it optimal projective codes}. Any line through the origin in $\F^d$, 
\begin{equation*}
    \label{eq:linemodel}
    x\F=\{x\lambda\ |\ \lambda\in \F\setminus\{0\}\}. 
\end{equation*}
 can be identified with a unit vector in $\F^{d}$ (one interesting fact is that line packings in $\mathbb{C}\mathbb{P}^1$ under this identification are in one-to-one correspondence with optimally spread collections of points on $\mathbb{S}^2$, we use this later). Through this correspondence, the chordal distance between points $ x,y \in \FP^{d-1} $ (lines in $\F^d$), has a simple formula
\begin{equation*}
\rho(x,y)=\sqrt{1-|\langle x,y \rangle|^2}.
\end{equation*}

 The problem of maximizing the minimal chordal distance is then equivalent to minimizing the {\it coherence} of a configuration, the maximal absolute value in the off-diagonal elements of the Gram matrix $\pmb{\Phi}^{*}\pmb{\Phi}$ of a set of unit vectors,  $\pmb{\Phi}=\{\pmb{\varphi}_{j}\}_{j=1}^N$, $|\pmb{\varphi}_{k}|=1$.  We use $\mu(\pmb{\Phi})$ to denote this quantity and $\S_{\F}^{d-1}$ to denote the unit sphere in $\F^{d}$. In searching for optimal codes then, the goal is to find the optimal values
$$\mu_{N,d}=\min\limits_{\pmb{\Phi}\subset \S_{\F}^{d-1}:\ |\pmb{\Phi}|=N}\max_{j\neq k}\  |\langle \pmb{\varphi}_{j}, \pmb{\varphi}_{k} \rangle|.$$

\section{Welch bound and equiangular lines}

One unusually symmetric type of packing of lines play a special role in the line packing problems. These equally spaced projective codes (which additionally must satisfy some conditions outlined below) are called {\it equiangular tight frames} (ETFs), {\it optimal simplices}, or {\it equiangular lines} \cite{LS}. The vectors in an ETF form a {\it tight frame}, and attain the {\it Welch bound}, a lower bound on the coherence of a configuration (detailed below) \cite{Wel}. When a tight frame is represented by a $d\times N$ matrix $\pmb{\Phi}$ over $\F$, it satisfies the property that the composition, $\pmb{\Phi\Phi}^{*}$ is a constant times the $d\times d$ identity matrix. Tight frames have similar properties to orthonormal bases, but are additionally overcomplete, or redundant. 

ETFs are precisely the systems of unit vectors which attain equality in the Welch bound, a lower bound on the minimal coherence of any code in $\FP^{d-1}$, 
$$\mu(\Phi)\geq \sqrt{\frac{N-d}{d(N-1)}}.$$
The size of an ETF cannot become too large in connection with Gerzon's bound, which says that this size is bounded above by the dimension of a corresponding space of symmetric matrices, $N\leq \binom{d+1}{2}, d^2$, or $2d^2-d$ in the real, complex, or quaternionic cases respectively. 

It is generally an open question for which dimensions maximally sized equiangular lines exist. Zauner's conjecture posits that in the complex case these configurations always exist, and this has been verified up to dimension $21$, while numerical evidence suggests the conjecture holds at least up to dimension $121$ \cite{Appl,Zau}. In the real case, a variety of methods show that equality in Gerzon's bound cannot hold in several dimensions, and it is known generally that $d+2$ must be the square of an odd integer if equality holds \cite{BMV}. In the quaternion case, there is a lack of numerical evidence for a similar phenomena which appears in the complex setting \cite{CKM}.

It is worth mentioning there is a type of dual construction, in which any ETF of size $N$ in $\F^{d}$ gives rise to another ETF of equal size in $\F^{N-d}$. One may check that by completing the $d\times N$ matrix to an appropriately scaled unitary matrix gives rise to another ETF. This is known as Gale duality and the new ETF is referred to as the Naimark complement of the first.

The Welch bound is only the first among several other lower bounds for the coherence of $N$ lines in $\F^d$. These include the orthoplex bound, the Levenshtein bound, and the Bukh-Cox bound \cite{BC}. It happens, although rarely when compared to Welch bound equality, that these bounds are met too for configurations, for instance by mutually unbiased bases and more rarely, the tight projective designs. For examples of the latter, take the tight projective designs of size $40$ and $126$ in $\CP^3$ and $\CP^5$ \cite{DGS,hoggar1982t}.

\section{Grassmannian packings and MISO}

By exploiting radio channel information MIMO allows for sending and receiving data simultaneously via multiple transmitting and receiving antennas. Complex optimal projective codes (also called Grassmannian packings) give good {\it beamforming codebooks} in MIMO but we focus below only on the specific case of MISO (Multiple-Input Single-Output) systems here due the result of \cite{LHS} which shows that codebook design for single stream transmission is independent of the number of receive antennas. 

The typical MISO system with $d$ transmit antennas and a single receive antenna can be modeled with a channel vector $\pmb{h}=[h_1,...,h_{d}]$ with independently and identically distributed (i.i.d.) complex Gaussian entries with zero mean and unit variance $\pmb{h}\sim\mathcal{N}_{\mathbb{C}}(\bf{0},\pmb{I}_{d})$. The channel vector captures transmission conditions (like physical obstructions in the path between antennas for instance) which result in the form of the received signal. The channel is often assumed to be static over a small time period (called the {\it coherence time}).

We transmit a symbol or bit $s$ over the channel using what is called a {\it beamforming vector} $\pmb{\varphi}$ where $\pmb{\varphi}=[\varphi^1,\varphi^2,\dots,\varphi^{d}]$ is a vector in $\mathbb{C}^{d}$ with unit norm. Altogether this gives us a simple model for the transmission of the symbol $s$ to the receiver $$y=\langle \pmb{\varphi},\pmb{h}s \rangle+\eta$$ where $\eta$ is complex Gaussian noise of mean zero and variance $\sigma$, and each of the $d$ entries of $\pmb{\varphi}^{\dagger}s$ are transmitted at the same time by the $d$ input antennas. 

A well studied method for choosing the beamforming vector in the above transmission model assumes that a codebook of potential beamforming vectors is agreed upon first and accessible to the transmitter and receiver $$\pmb{\Phi}=[\pmb{\varphi}_1,\pmb{\varphi}_2,\dots,\pmb{\varphi}_{N}].$$ Further it is assumed the receiver can measure in an error-free manner $$\max\limits_{i=1,\dots,N} |\langle \pmb{\varphi}_{i}, \pmb{h} \rangle|^2$$ and may select a $\pmb{\varphi}_i$ which maximizes this quantity, relaying the relevant index to the transmitter before it sends the symbol $s$. 

At this point the transmitter uses the code word $\pmb{\varphi}_i$ to transmit the symbol $s$. This choice maximizes the signal-to-noise ratio (SNR) within a small time interval of the receiver's measurements  $$\gamma=\frac{|\langle \pmb{\varphi},\pmb{h}\rangle|^2 E_{s}}{\sigma}$$

\noindent where $E_s=\mathbb{E}\left[|s|^2\right]$ is the average symbol error. Note that the value of $s$ can now be trivially estimated from $y$ via dividing by $\langle \pmb{\varphi},\pmb{h} \rangle $ \cite{ZWG}.

A good beamforming codebook in this setting unsurprisingly is found by taking $\pmb{\Phi}$ to be a complex line packing \cite{LHS}. In other words if $\pmb{\Phi}$ minimizes coherence over all systems of $N$ unit norm vectors in $\mathbb{C}^{d}$ (or even approximately) then it serves as a good set of vectors for the MISO transmission application above. 

To see this, let $d_{c}$ be the chordal distance, and $h\sim \sigma_{\mathbb{CP}^{d-1}}$ be uniformly distributed on the space of complex lines. The quantization scheme that corresponds to the optimal choice of $\pmb{\varphi}_i$ above then takes the form 
$$ \mathcal{Q}_{\pmb{\Phi}}(\pmb{h}):\pmb{h}\mapsto \arg\max\limits_{i=1,\dots,N} |\langle \pmb{\varphi}_{i}, \pmb{h} \rangle|^2=\arg\min\limits_{i=1,\dots,N} d_{c}(\pmb{\varphi}_{i}, \pmb{h})$$
so that a codebook $\pmb{\Phi}$ has distortion metric 
$$\mathcal{D}(\pmb{\Phi})=\mathbb{E}[d_c^2(\sigma_{\mathbb{CP}^{d-1}},\mathcal{Q}_{\pmb{\Phi}}(\sigma_{\mathbb{CP}^{d-1}}))].$$

This distortion is hard to express generally. One of the main results of \cite{LHS} is that this distortion measure may be bounded by an increasing function of the coherence of the codebook $\pmb{\Phi}$ and so as a proxy for designing a codebook which minimizes the distortion, we find one minimizing coherence.

\section{Optimization procedure}

We now introduce a new optimization procedure we use to optimize for low coherence matrices or complex Grassmannian packings which we call Trust-Region Stochastic Tuning for Matrix Incoherence (TRSTMI). This procedure uses an alternating method of computing successive smooth approximations to the maximum function (log-sum-exp) coupled with an unconstrained minimization procedure (trust-region conjugate gradient method, a popular nonlinear optimization method \cite{nowr}) used to minimize the smooth maximum function approximation evaluated on the absolute value squared of inner products of a unit norm projection of a given matrix. A random Monte-Carlo sampling procedure is additionally employed to address the fact that many trials may reach a locally optimal value of the coherence which is far from the global minimum.

The maximum approximation mentioned above takes the form $$F_{s,\delta}(\pmb{x})=s+\delta\log{\left[\sum_{i=1}^n \exp{\left(\frac{x_i-s}{\delta}\right)}\right]}$$ where $s=\max\limits_{i} x_i$ and $\delta>0$ is small. The method of optimizing successive approximations of the max function evaluated on inner products was investigated also in \cite{ARU}.  The reference \cite{ZB} gives a good survey of other numerical approaches for contructing line packings.

We choose a series of decreasing $\delta_{k}'s$ (which are experimentally chosen based on performance and guided by elementary considerations) and feed the optimized matrix $\pmb{x}'$ which comes out of minimization of $F_{s,\delta_{k}}(\pmb{x})$ into a minimization of the next approximation $F_{s,\delta_{k+1}}(\pmb{x}')$. Eventually we reach a terminally small value of $\delta$ which decreasing further tends not to result in any considerate improvement in the optima quality (for similar time spent). The variable $\pmb{x}$ here corresponds to the upper triangular part of the absolute value of the matrix of inner products obtained after normalizing the columns of a matrix to be unit norm (it suffices to only compute the maximum of the upper-triangular part by symmetry).

\begin{algorithm}\caption{Trust-Region CG-Steihaug  \cite{nowr}[Alg. 7.2]}
	\label{alg:CG}
	\alglanguage{pseudocode} 
	\begin{algorithmic}
		\State $z_0 \gets 0,\ r_0 \gets \nabla{f}_k,\  d_0 \gets -r_0=-\nabla f_k$
		\If{$\|r_0\|<\epsilon_k$}
		\Return{$p_k\gets z_0=0$}
		\For{$j=0,1,2,...$}
		\If{$d_{j}^TB_kd_j \leq 0$} 
		\State{Find $\tau$ s.t. $p_{k}=z_j+\tau d_j$ minimizes (\textdagger)} 
		\State \Return{$p_k$;}
		\EndIf
		\State $\alpha_{j} \gets r_j^{T}r_{j}/d_{j}^{T}B_kd_{j}$
		\State $z_{j+1} \gets z_j+\alpha_j d_j$
		\If{$\|z_j\| \geq \Delta_k$}
		\State Find $\tau\geq 0$ s.t. $p_{k}=z_j+\tau d_j$ and $\| p_k \|=  \Delta_k$
		\State \Return{$p_k$;}
		\EndIf
		\State $r_{j+1} \gets r_{j}+\alpha_j B_k d_j$
		\If{$\|r_j \| < \epsilon_k$}
		\State \Return{$p_k \gets z_{j+1}$};
		\EndIf
		\EndFor
		\EndIf
		\State (\textdagger): $\min_{p\in\mathbb{R}^n} m_k(p)=f_k+\nabla f_k^{T}p+\frac{1}{2} p^{T}B_k p$ subject to \\ $\|p\|\leq \Delta_k$, where $B_k=\nabla^2 f_k$
		\end{algorithmic}
\end{algorithm}

We now show how an elementary inequalities suggest a choice of $\delta$ which gives good approximation to the maximum function.

\begin{proposition} Let $F_{s,\delta}(\pmb{x})=s+\delta\log{\left[\sum_{i=1}^n \exp{\left(\frac{x_i-s}{\delta}\right)}\right]}$, where $s=\max\limits_{i} x_i$. Then, $$s\leq F_{s,\delta}(\pmb{x})\leq s+\delta\log{n}.$$
\end{proposition}

\begin{proof}
The first (left) inequality is trivial since $\delta>0$ and $$\sum_{i=1}^n \exp\left(\frac{x_i-s}{\delta}\right) \geq \exp\left(\frac{s-s}{\delta}\right)=1.$$ Let $y_i=\frac{x_i-s}{\delta}$, then $$\log\left[\sum_{i=1}^n \exp(y_i)\right]\leq \log\left[n\cdot\exp(\max\limits_{i} y_i)\right].$$ Thus, since $\max\limits_{i} y_i=0,$ we have  $$F_{s,\delta}(\pmb{x})=s+\delta\log\left[\sum_{i=1}^n \exp\left(\frac{x_i-s}{\delta}\right)\right]\leq s+\delta\log{n}.$$
\end{proof}

\begin{algorithm}\caption{TRSTMI optimization steps}
\begin{enumerate}
    \item generate random collection of $N$ unit vectors. 
    \item run optimization procedure
    \begin{enumerate}
        \item approximate max function with log-sum-exp approximation $F_{s,\delta_{k}}(\pmb{x})$ (with $\delta_k\rightarrow 0$ decreasing sequence).
        \item apply trust-region conjugate gradient method to minimize $F_{s,\delta_k}(\pmb{x})$ where $\pmb{x}$ is the vectorized off-diagonal Gram matrix entries (corresponding to the unit vectors).
        \item initialize optimization of $F_{s,\delta_{k+1}}$ with the optimizer of $F_{s,\delta_k}(\pmb{x})$.
        \item repeat for all values of $\delta_{k}$, $k=1,...,M$.
    \end{enumerate}
    \item repeat procedure for $L$ random initializations and return the minimum over these trials.
\end{enumerate}
\end{algorithm}

\begin{table}[h] 
    \begin{center}
        \caption{A comparison of time needed to optimize for a single step (a single approximation of max function) for different large values of $(d,N)$ for CPU and GPU parallelized procedures. The GPU runs were completed on an Nvidia A100 80 GB RAM graphics processing unit. The CPU runs were done on either a single or 6 parallel cores on a Ryzen 5 3600 CPU. }
        \label{tab:gpucpu}
        \begin{tabular}{ c  c  c  c  c  } 
            $d$  & $N$ & GPU $(s)$ & CPU - 6 cores $(s)$ & CPU - 1 core $(s)$ \vspace{1 mm} \\ \hline 
\rule{0pt}{3ex} 

$4$ & $18000$ & $1161$ & $2927$ & $4075$   \\

\ $5$ & $18000$ & $1918$  & $3896$ & $5661$   \\

\ $6$ & $18000$ & $3400$  & $3641$ & $10128$  \\ 
   
        \end{tabular}
    \end{center}
\end{table} 

\begin{table}[h] 
    \begin{center}
        \caption{A comparison of time needed to optimize for different large values of $(d,N)$ for different procedures (all on CPU). The TRSTMI trials and MATLAB implementations of Alt-Proj and CBGC were run on an i7-8750H processor. Values are rounded up at fourth decimal to give actual upper bounds (for example $.50001$ would be rounded to $.5001$ in the table).}
        \label{tab:comp}
        \begin{tabular}{ c  c  c  c  c} 
            $d$  & $N$ & TRSTMI $(s)$ & Alt-Proj $(s)$ &  CBGC $(s)$  \vspace{1 mm} \\
            $$  & $$ & ($\mu$)  & ($\mu$) & ($\mu$)  \vspace{1 mm} \\ \hline
\rule{0pt}{3ex} 

$2$ & $4$ & $2.36$ & $0.613$  & $0.0156$ \\
 &  & $(0.5774)$ & $(0.5774)$ &  $(0.5774)$ \\

\ $3$ & $9$ & $13.1$  & $1.20$ &  $4.34$  \\
 &  & $(0.5001)$  & $(0.5001)$  &  $(0.5001)$  \\

\ $4$ & $16$ & $21.3$  & $2.73$ &  $0.301$ \\ 
 &  & $(0.4473)$  & $(0.4473)$  &  $(0.4473)$  \\

\ $5$ & $25$ & $39.3$  & $10.9$  &  $0.644$ \\
 &  & $(0.4083)$  & $(0.4083)$   &  $(0.4083)$ \\

\ $6$ & $36$ & $79.7$  & $19.5$  & $1.37$ \\ 
 &  & $(0.3780)$  & $(0.3780)$  & $(.3780)$ \\

\ $7$ & $49$ & $146$  & $25.6$ & $859$  \\ 
 &  & $(0.3536)$  & $(0.4545)$  & $(.3848)$ \\

\ $8$ & $64$ & $189$ & $48$  & $1837$ \\
& & $(0.3705)$ & $(.4782)$  & $(0.3693)$  \\

\ $9$ & $81$ & $388$ & $61.1$  & $-$ \\
& & $(0.3575)$ & $(.4648)$ & $-$ \\
   
        \end{tabular}
    \end{center}
\end{table}

\begin{remark} In our case $\pmb{x}$ is the vectorized form of the upper part of a square $N\times N$ matrix and hence $n=\frac{N(N-1)}{2}$. Hence to get $F_{s,\delta}$ within $\epsilon$ of the value of the max function we may take $\delta\approx \frac{\epsilon}{2\log{N}}$. 
\end{remark}

\begin{remark}

Without the shift by $s$ within the exponential function, evaluation of the inner function will `blow up' for small $\delta$. We handle differentiation of nonsmooth $s$ via the subgradient method \cite{Nes}, so that in our gradient based minimization procedure we compute a subgradient as a proxy for the gradient of the non-smooth function $\max\limits_{i} x_i$. 
\end{remark}

\begin{remark} We are able to use an unconstrained minimization technique because we divide each column of the $d\times N$ matrix representing configuration $\pmb{\Phi}$ by the norm of the associated column vectors (precluding a zero vector column from our matrix). One way to improve performance might be to use another representation of the unit norm columns (for instance using manifold optimization such as in \cite{cbs}). 
\end{remark}

\begin{remark} We implement TRSTMI in PyTorch\footnote{Github repository with our code: \url{https://github.com/JosiahPark/TRSTMI}.} \cite{PGM} so as to allow for straightforward parallelization support on CPU/GPU. We found that this led to notable decreases in runtime for large problems and detail the result of some experiments in Table~\ref{tab:gpucpu}. In Table~\ref{tab:comp} we compare the runtime and performance of TRSTMI to other established methods for constructing complex line packings \cite{TDH,DL} (because these other methods are implemented in MATLAB instead, some care should be applied in comparing runtimes). In these experiments we stuck to $L=2d^2$ random initializations with TRSTMI. We note that CBGC is an energy based optimization procedure and that some properties of minimizers of associated energies are derived in \cite{gp}.
\end{remark}

\begin{remark} One approach which has been successful for finding high precision optimal spherical codes (best packings of spherical caps on the surface of a sphere) involves successive minimization of harmonic energies \cite{Wang1,Wang2}, an adjustment of an algorithm from \cite{kottwitz}. We find this procedure is very time consuming for large problems, but expect that as a (likely expensive) post-processing step, a singular energy optimization procedure may be applied to the output of TRSTMI to possibly further decrease coherence.
\end{remark}

\begin{remark} In our experiments we found that TRSTMI improved on many existing estimates for $\mu_{N,d}$ from the online database associated with \cite{JKM}. We collected a few that we expect will be harder to improve upon in Table~\ref{tab:comp2}.
\end{remark}

\section{Conjectures from Numerics}

A well-known phenomena that appears when studying best packing problems in compact spaces is that one can obtain other seemingly optimal packings from 'exceptional' packings by removing one or more points. This is not unique to packings of subspaces, but is a phenomena which appears in a variety of geometric packing questions. 

One example of this phenomena can be understood via the {\it orthoplex bound}, a coherence lower bound similar to the Welch bound. By removing points from {\it mutually unbiased bases} (MUBs) the orthoplex bound shows one often obtains an optimal packing.

\begin{proposition}[Orthoplex and Levenshtein bound, \cite{CHS,L}]
 \label{thm:bounds}
Fix $d>1$, select $\Phi \in \mathbb{C}^{d\times N}$ with unit norm columns, and set $\mu = \mu(\Phi)$. Then \\ 
\begin{equation*}
\label{eq:orth}
\ \mu(\Phi) 
\geq \frac{1}{\sqrt{d}},\ \ \textrm{if $n > d^2$,} \\ \vspace{2 mm}
\end{equation*}
\begin{equation*}
\label{eq:Lev}
\ \mu(\Phi) 
\geq \sqrt{\frac{2n-d(d+1)}{(n-d)(d+1)}},\ \  \textrm{if $n>d^2$.} 
\end{equation*}

Equality in the orthoplex bound~\eqref{eq:orth} requires $n \leq 2d^2-1$, and equality in the Levenshtein bound~\eqref{eq:Lev} requires that $\Phi$ be tight with angle set $\{0,\mu^2\}$.

\end{proposition}

Two orthonormal bases in $\CP^{d-1}$ are said to be {\it mutually unbiased}  if they are aligned so that the inner product of any pair of vectors, $\varphi_{j},\varphi_{k}$, each from distinct bases, satisfies
$$|\langle \varphi_{j},\varphi_{k} \rangle|=\frac{1}{\sqrt{d}}.$$ \indent It is known that the maximal number of MUBs in $\CP^{d-1}$ is no larger than $d+1$, and that for prime power dimension, $d=p^k$, it is known this maximum number is attained. Interestingly, it is an open question what is the maximal number of MUBs in dimension $d=6$, where it is conjectured that at most $3$ exist. 

An immediate consequence of the orthoplex bound is that for $d=p^k$, where $p$ prime, any $(d+1)$-MUB gives rise to an optimal projective packing of $N=(d+1)d-j$ points, $j=0,\dots,d-1$ by removal of any $j$ points in the $(d+1)$-MUB. In particular, $\mu_{N,d}=\frac{1}{\sqrt{d}}$ for $d$ and $N$ as above. It is an interesting question of whether such configurations are the unique coherence minimizers.

Our first conjecture suggests a similar phenomena occurs for maximal complex equiangular tight frames (also known as SIC-POVMs \cite{RBSC}).

\begin{conjec} \label{thm:sic} For $d\geq 2$, there exists a SIC of $N=d^2$ points and for each $N=d^2-j$, $j=0,\dots d-2$, each optimal projective code of size $N$ is given by the removal of some $j$ points from the SIC. In particular, $\mu_{N,d}=\frac{1}{\sqrt{d+1}}$ for $d$ and $N$ as above.
\end{conjec}
The first statement on existence is just Zauner's conjecture. However, if the above statement is true, this would imply non-existence of an ETF of size $N$, for $N$ slightly smaller than $d^2$ as above, since 
$$\frac{1}{\sqrt{d+1}}>\sqrt{\frac{d^2-j-d}{d(d^2-j-1)}},$$
holds for any $j\in \{1,\dots,d-2\}$. For comparison, the only non-existence result for complex ETFs (excluding those ruled out by Gerzon's bound) shows for $d=3$ there does not exist an ETF of size $N=8$, a special case of the above conjecture. Conjecture~\ref{thm:sic} also appears for the specific case of $d=3$ in \cite{JKM}.

It is interesting how sharp the threshold is for removing points while still obtaining an optimal packing. We expect that for $d=2,3,4,$ and $5$ the above conjecture gives the right threshold. Besides maximally sized MUBs or SICs there are several other exceptional configurations for which it appears one can remove a point, still obtaining an optimal packing. In increasing dimension, some examples from numerics appear to be the tight $3$-design in $\CP^{3}$ of $40$ points, the $W(K_5)$ line system of $45$ points in $\CP^{4}$, an ETF of $N=31$ points in $\CP^{5}$, and an ETF of $28$ points in $\CP^6$.

We expect that the $126$ point tight $3$-design also has the above property. A newly found highly symmetric, small sized projective $3$-design (described in full in \cite{pa}) of $85$ points in $\CP^4$ is also expected to have this property. In all of the above cases, it is not clear how many points can be removed so as to obtain an optimal packing, nor for what reason seeming differences in this number occur for similar configurations.

One phenomena which is displayed numerically is that small complex grassmannian packings appear to have a single distinct distance between lines for $N$ up to some threshold $m(d)$ (which appears to increase non-trivially with the dimension $d$). Since a one distance set cannot become too large in connection with classical bounds, there is a maximum size $m(d)\leq d^2$ for which there exists a packing of size $N\in[d,m(d)]$ with only a single distinct distance appearing. By default $m(d)\geq d+1$, since the orthonormal basis and the simplex are optimal packings.

From our optimal candidates in dimensions $d=2$ through $8$ and for increasing $N$, there is a value of $N$ where numerically optimal configurations go from having very small differences in the absolute value of their off-diagonal Gram matrix entries to noticeably large differences. We tabulate our observations in this direction in Table \ref{table:md}.

\begin{table}[h]
    \begin{center}
        \caption{Consecutive small values of $N$ for which it appears an optimal packing in $\CP^{d-1}$ of size $N$ can be taken to have a single distance appearing.}
        \begin{tabular}{ c c }
            \label{table:md}
$d$ & $N$ \\ \hline
\rule{0pt}{3ex}2 & 2--3  \\  
3 & 3--4  \\  
4 & 4--8  \\  
5 & 5--13 \\  
6 & 6--16  \\  
7 & 7--19  \\  
8 & 8--22  \\  

\end{tabular}
    \end{center}
\end{table} 

\begin{table}[h] 
    \begin{center}
        \caption{A comparison of select coherence values (rounded up at sixth digit) obtained via TRSTMI to established constructive upper bounds from \cite{JKM}. }
        \label{tab:comp2}
        \begin{tabular}{ c  c  c  c   } 
            $d$  & $N$ & TRSTMI coh. & \cite{JKM} database coh.  \vspace{1 mm} \\ \hline
\rule{0pt}{3ex} 

3 & 27 & 0.734233 & 0.73726116 \\
\ 3 & 28 & 0.737797 & 0.73884638 \\
\ 6 & 24 & 0.371529 & 0.37267800 \\

        \end{tabular}
    \end{center}
\end{table}

\begin{figure}[H]
{\centering
\includegraphics[width=.35\textwidth]{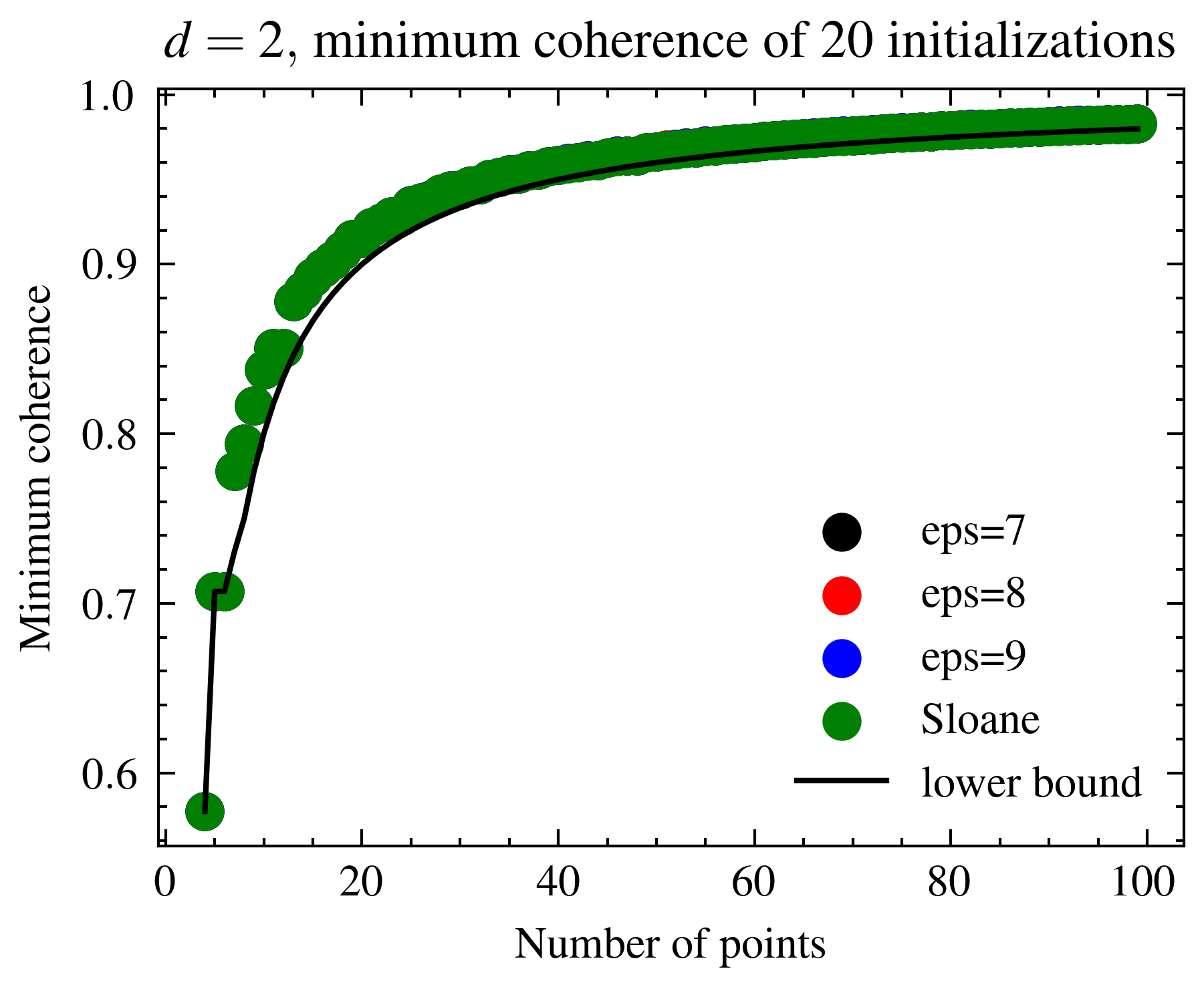} \\ 
\includegraphics[width=.35\textwidth]{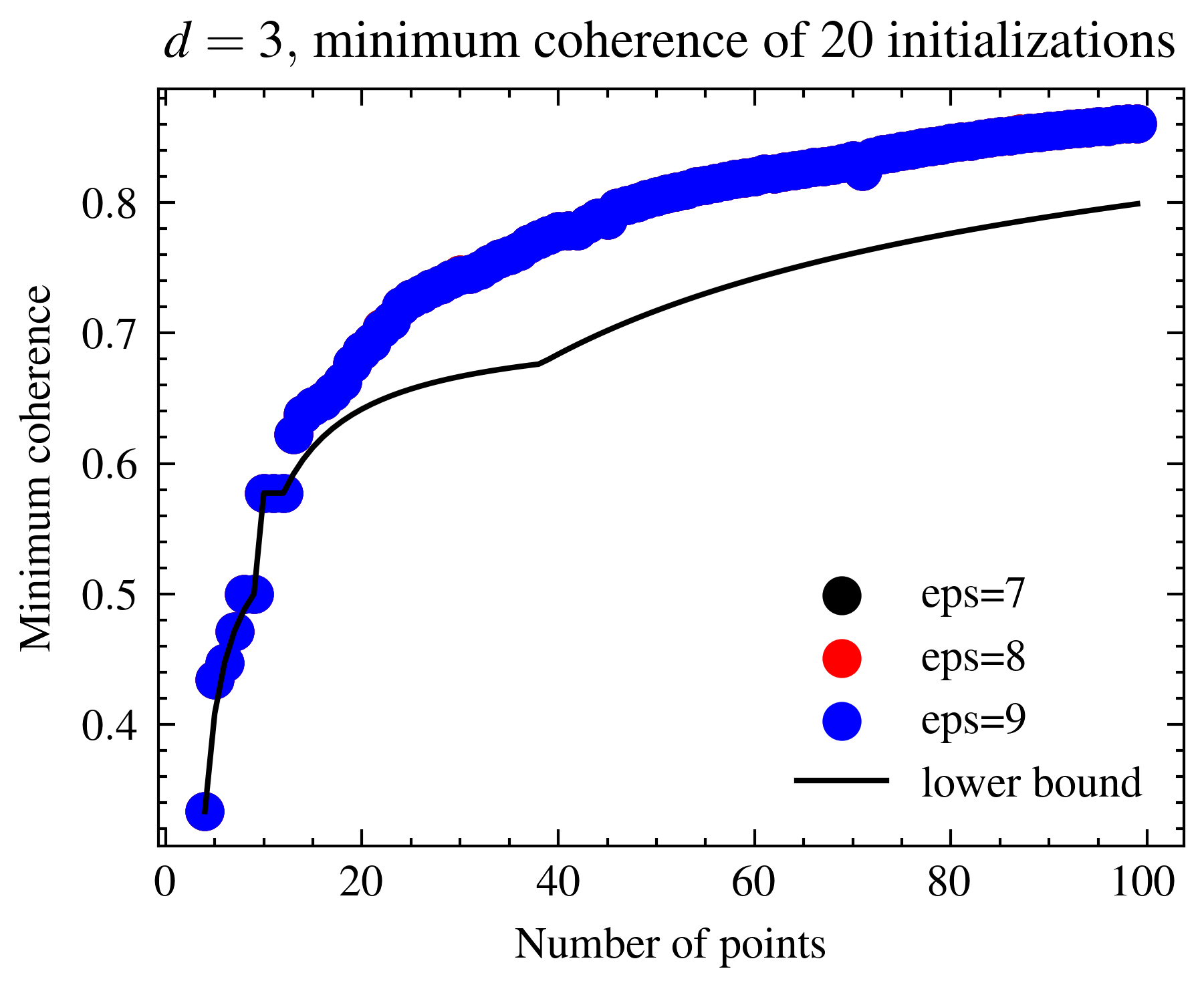} \\
\includegraphics[width=.35\textwidth]{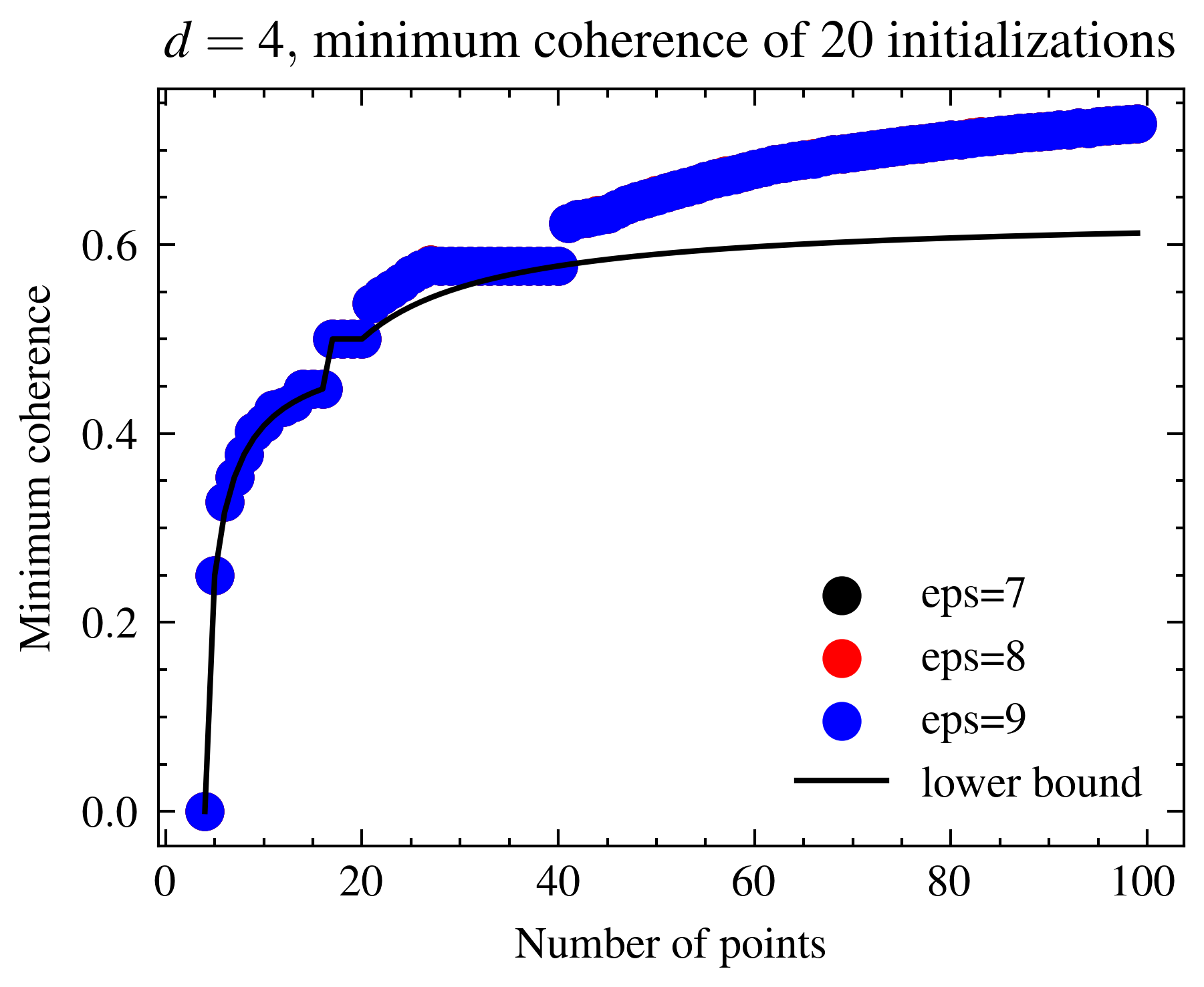} \\}
\caption[LoF entry]{Comparison of lower bounds for $\mu_{N,d}$ against constructive numerical upper bounds for $N\leq 100$ and $\mathbb{C}\mathbb{P}^{1}$, $\mathbb{C}\mathbb{P}^{2}$, and $\mathbb{C}\mathbb{P}^{3}$ obtained via the TRSTMI optimization procedure (minimum values among 20 random intializations for each set of parameters). With log-sum-exp approximation level `eps'$=7,8,9$ for the maximum function (code specific parameter corresponding to $\delta=10^{-7}-10^{-8}$) and the scale taken in the plots, the difference with the values coming from Sloane's database for optimal spherical codes are indistinguishable \cite{sloane}. 

\indent Some successive coherence values are smaller than those preceding. Since removing a point from a larger configuration does not increase the coherence, the plots can be improved in some cases trivially. We show the raw optimized values instead. \label{fig:exp}}

\end{figure}

\section{Acknowledgements}

The first author would like to acknowledge AIM for providing aid to attend the `Discrete Geometry and Automorphic Forms' organized by Henry Cohn and Maryna Viazovska where work on this project began. Portions of this research were conducted with the advanced computing resources provided by Texas A\&M High Performance Research Computing.

The authors would also like to acknowledge helpful correspondence from Henry Cohn, Simon Foucart, John Jasper, Emily J.~King, Dustin G.~Mixon, and Oleksandr Vlasiuk.






\ifCLASSOPTIONcaptionsoff
  \newpage
\fi

\end{document}